\documentclass{article}
\usepackage[a4paper]{geometry}

\usepackage[utf8]{inputenc}
\usepackage{authblk}

\usepackage{xcolor}
\usepackage[spaces,hyphens]{url}
\usepackage{hyperref}
\usepackage{graphicx}
\usepackage{adjustbox}
\usepackage{comment}

\usepackage{xspace}
\usepackage{algorithm}
\usepackage[noend]{algpseudocode}
\usepackage{amsmath,amsthm,amssymb}
\newtheorem{lemm}{Lemma}
\usepackage{amssymb,amsfonts}
\usepackage{textcomp}

\let\olditemize\itemize
\renewcommand\itemize{\olditemize\setlength\itemsep{1pt}\parskip0pt}
\let\oldenumerate\enumerate
\renewcommand\enumerate{\oldenumerate\setlength\itemsep{1pt}\parskip0pt}

\usepackage[backend=bibtex, style=numeric, sorting=none, giveninits=true]{biblatex}
\begin{document}
	
	\title{Clover: an Anonymous Transaction Relay Protocol\\for the Bitcoin P2P Network\thanks{This work was supported by Project RTI2018-102112-B-I00 (AEI/FEDER,UE).} 
	}
	
	
	\author[]{Federico Franzoni} 
\author[]{Vanesa Daza} %
\affil[]{Universitat Pompeu Fabra, Barcelona, Spain\\
\texttt{\{federico.franzoni,vanesa.daza\}@upf.edu}
}
	\date{}

	\maketitle
	
	\begin{abstract}
		The Bitcoin P2P network currently represents a reference benchmark for modern cryptocurrencies.
		Its underlying protocol defines how transactions and blocks are distributed through all participating nodes.
		To protect user privacy, the identity of the node originating a message is kept hidden.
		However, an adversary observing the whole network can analyze the spread pattern of a transaction to trace it back to its source.
		This is possible thanks to the so-called \textit{rumor centrality}, which is caused by the symmetry in the spreading of \textit{gossip}-like protocols.
		
		Recent works try to address this issue by breaking the symmetry of the Diffusion protocol, currently used in Bitcoin, and leveraging proxied broadcast.
		Nonetheless, the complexity of their design can be a barrier to their adoption in real life.
		In this work, we propose Clover, a novel transaction relay protocol that protects the source of transaction messages with a simple, yet effective, design.
		Compared to previous solutions, our protocol does not require building propagation graphs, and reduces the ability of the adversary to gain precision by opening multiple connections towards the same node.
		Experimental results show that the deanonymization accuracy of an \textit{eavesdropper} adversary against Clover is up to 10 times smaller compared to Diffusion.
	\end{abstract}
	
	\section{Introduction}
	Over the past few years, Bitcoin \cite{nakamoto2008peer} has risen to an unprecedented level of popularity.
	Although many users still believe this system to be anonymous, studies showed how it is relatively easy to link transactions to real identities~\cite{androulaki2013evaluating,reid2013analysis,herrera2015research}.
	Most deanonymization techniques work by tracing transactions on the blockchain and combining them with publicly available knowledge~\cite{meiklejohn2013fistful,nick2015data,neudecker2017could}.
	
	However, a less-known approach is to link transaction messages to their originating node in the underlying P2P network \cite{biryukov2014deanonymisation,koshy2014analysis,biryukov2015tor}.
	This approach is based on the observation that the first device to broadcast a transaction in the network is likely the one that created it.
	To implement this approach, an adversary typically deploys one or more nodes connecting to all reachable peers in the network, and listens for incoming transaction messages~\cite{fanti2017deanonimization}.
	Transactions are then linked to their source by using an estimation strategy.
	This type of adversary is known as the \textit{eavesdropper} adversary.
	
	A recent work by Fanti et al.~\cite{fanti2017anonymity} shows that the Diffusion protocol, currently used in Bitcoin, has poor anonymity guarantees against this adversary.
	In particular, an attacker can obtain high levels of precision even when controlling just few nodes in the network and using a naive estimation strategy.
	The authors identify the problem in the symmetry of the spreading pattern: since transactions spread from each node to all its peers, it is always possible to determine the approximate point in the network where the propagation started.
	This phenomenon is known as \textit{rumor centrality} and is specific to all gossip-like systems~\cite{shah2012rumor,shah2011rumors}.
	
	Following these findings, few solutions have been proposed that reduce the ability of the adversary to identify the source of a transaction \cite{venkatakrishnan2017dandelion,franzoni2020improving}.
	These proposals break the symmetry in the propagation pattern by having nodes delegate the broadcast of new transactions to other nodes of the network.
	In particular, transactions are first propagated (or \textit{proxied}) linearly over a path of nodes, and then broadcast using the Diffusion protocol.
	The way nodes are chosen during this initial phase is defined by the protocol, and determines the security and complexity of the solution.
	In Dandelion \cite{venkatakrishnan2017dandelion}, reachable nodes in the network build a propagation graph passing through every node (i.e., an Hamiltonian Circuit) and always propagate new transaction over the same path. Given the risk of the adversary learning the topology of such graph, a new graph has to be built periodically.
	In \cite{franzoni2020improving}, the initial phase alternates reachable and unreachable nodes with the goal of concealing the propagation process. Since the adversary can control multiple unreachable peers for any given node, the authors suggest the use of \textit{bucketing} \cite{bucketing} to mitigate her ability of tracking transactions.
	
	While these solutions sensibly improve the anonymity properties of transaction propagation, their adoption is hindered by their complexity.
	Additionally, in both protocols, the adversary can gain an advantage by learning sensitive information on the initial phase, such as the nodes in the propagation path (in Dandelion) or the transactions being proxied (in \cite{franzoni2020improving}).
	
	In this paper, we propose a new propagation protocol that breaks the symmetry by separating inbound and outbound connections in the relay pattern.
	Our design is simpler than previous solutions, making its analysis and implementation easier.
	Additionally, we minimize security concerns by limiting the ability of the adversary to learn sensitive information: our protocol does not require build a propagation graph, and prevents the adversary from tracking transactions by isolating inbound connections.
	We formally analyze our protocol and experimentally evaluate it by using a proof of concept in a simulated environment.
	Our results show that, compared to Diffusion, our protocol reduces the deanonymization precision for the eavesdropper adversary from 0.6 to just 0.05, in the best case, and from 0.7 to 0.3 in the worst case.

	\section{Background}
	\label{sec:spreading}
	
	\subsection{The Bitcoin P2P network}
	The Bitcoin P2P network is composed of nodes randomly connected among each other.
	Peers of a node are distinguished between outbound, whose connection was opened by the node, and inbound, from which the node accepted an incoming connection.
	According to the Bitcoin reference client, each node establishes and maintains 8 outbound connections and, if reachable, up to 117 inbound connections.
	Thus, reachable nodes can have up to 125 connections, while unreachable nodes are limited to 8.
	
	However, this limit is not enforced, making nodes able to establish as many connections as needed.
	This is particularly useful for measuring tools that connect to all reachable nodes, as well as for the so-called \textit{supernodes}, which are often used by mining pools to maximize their connectivity with the network.
	At the same time, malicious actors can exploit this feature to improve the effectiveness of their attacks.
	
	\subsection{Transaction propagation}
	\label{sec:txprop}
	A transaction \textit{tx} is transmitted from a node \textit{A} to a node \textit{B} following a three-step process:
	\begin{enumerate}
		\item Node \textit{A} announces \textit{tx} to node \textit{B} by sending an \texttt{INV} message, containing the hash of the transaction (\textit{h(tx)});
		\item If \textit{h(tx)} is unknown, node \textit{B} requests \textit{tx} to node \textit{A} by sending a \texttt{GETDATA} message is sent, containing \textit{h(tx)};
		\item Node \textit{A} sends \textit{tx} to node \textit{B} via a \texttt{TX} message.
	\end{enumerate}
	This announcement-based propagation mechanism is used to avoid transmitting a transaction twice to the same node.
	To optimize network data consumption, \texttt{INV} messages usually aggregate multiple transaction hashes.
	
	Transactions are spread from a node to its peers following a gossip-like protocol known as \textit{Diffusion}, which works this way: 
	when a new transaction is created or received by a peer, it is announced to all connected peers; 
	before sending the \texttt{INV} message, an individual random delay is applied to each peer.
	
	\subsection{Deanonymization strategies}
	In Bitcoin, the broadcast and relay operations follow the same rules.
	Therefore, when nodes create a transaction, they propagate it the same way as transactions received from their peers.
	This approach is used to prevent leaking the identity of the node that originates the transaction.
	However, as we already mentioned, it is possible to determine the source of a transaction
	by observing its propagation through the network.
	
	In particular, an eavesdropper adversary, which connects to all reachable nodes in the network, can adopt different strategies to estimate the source of a transaction. 
	The simplest method, called \textit{first-timestamp} or \textit{first-spy} estimator, consists in linking each transaction to the first node that announces it (to the adversary).
	The rationale behind this method is that the first node to announce a transaction in the network is likely the one that generated it.
	By connecting to all nodes, the adversary is always likely to receive each transaction from its source.
	This strategy has been proved to reach very high levels of accuracy against Diffusion, even when the adversary controls only few nodes~\cite{fanti2017anonymity}
	
	More advanced techniques are theoretically possible when the adversary knows the network topology \cite{fanti2017anonymity}.
	These techniques take into account the propagation of transactions to exploit the rumor-centrality property of the Diffusion protocol.
	In particular, these techniques are based on the observed order in which nodes announce the transaction.
	The underlying assumption is that nodes that are (topologically) closer to the source will announce the transaction earlier than those that are farther away.
	Although these methods are theoretically more precise than the first-spy estimator, their adoption is conditioned to the knowledge of the network graph, which is currently obfuscated by the Bitcoin protocol.

	\section{Adversarial Model}
	\label{sec:adversary}
	We consider the eavesdropper adversary defined in \cite{fanti2017anonymity}, which is based on practical attacks such as \cite{biryukov2014deanonymisation} and \cite{koshy2014analysis}.
	This adversary makes use of a supernode that connects to all reachable nodes in the network.
	For each node, multiple connections can be established by using different IP:port addresses, making them look as coming from different entities.
	In particular, the adversary can fill up all unused inbound slots of a target node.
	
	Additionally, we extend this adversary by letting it deploy an arbitrary number of reachable nodes 
	with the objective of being selected as an outbound peer by other nodes.
	This extension allows the adversary to improve precision against our protocol.
	
	The goal of the adversary is to determine the source node for all received transactions.
	To this purpose, adversarial nodes listen to all messages relayed by their peers, logging their content and timestamp.
	We assume the adversary has a unified view of the logs from all the nodes under its control.
	Furthermore, for each honest node, the adversary maintains a \textit{deanonymization set}, which contains all the transactions that have been possibly generated by that node.
	To deanonymize transactions, the adversary adopts the first-spy estimator: each transaction is linked to the first node that announces it (to any of the nodes controlled by the adversary).

	\section{The Clover protocol}
	In this section, we describe Clover, our new transaction propagation protocol.
	We detail and motivate its design with reference to the adversary model.
	
	\subsection{Protocol Overview}
	Similar to previous solutions \cite{venkatakrishnan2017dandelion,franzoni2020improving}, Clover protects the source of a transaction by means of \textit{proxying}.
	This consists in delegating the broadcast of new transactions to other nodes.
	Specifically, when a node creates a new transaction, it selects one of its peers and sends it the transaction.
	The selected node, called \textit{proxy}, is then responsible for broadcasting the transaction to the rest of the network.
	Proxying allows moving the apparent origin of the propagation of a transaction from its source to a different node of the network.
	
	Note that proxying drastically reduces the effectiveness of the first-spy estimator approach, since it is highly unlikely for the eavesdropper adversary to receive a transaction from its source.
	However, if the adversary controls the selected proxy, she could be able to distinguish a proxied transaction and simply link it to the sender node (i.e., deanonymize it).
	
	To mitigate this risk, we use transaction \textit{mixing}.
	This consists in making nodes proxy their new transactions along with transactions created (and proxied) by other nodes.
	This strategy reduces the ability of an adversarial node of determining whether a proxied transaction was created by the sender or a different node.
	In particular, the more the transactions used for mixing, the lower the precision of the adversary.
	We call \textit{mixing set} the set of transactions used by a node for mixing.
	
	To enable mixing, new transactions are proxied over multiple nodes before being broadcast (multi-hop proxying).
	In other words, a new transaction propagates in two phases: the \textit{proxying phase}, in which the transaction is relayed through a number of proxy nodes, and the \textit{diffusing phase}, where the transaction is broadcast and propagated following the Diffusion protocol.
	The switch between the proxying phase and the diffusing phase can occur at any hop, and it is determined probabilistically by the node that receives it.
	Specifically, when a node receives a transaction in the proxying phase, it decides whether to relay it to another proxy (thus keeping it in the proxying phase) or broadcast it with Diffusion (thus switching to the diffusing phase).
	We call \textit{proxy transaction} a transaction in the proxying phase, and \textit{diffused transaction} a transaction in the diffusing phase.
	
	To improve anonymity, only proxy transactions are used for mixing.
	This is motivated by the observation that diffused transactions are likely to be known by the adversary.
	In fact, by connecting to all reachable nodes, an eavesdropper adversary is always among the first to receive such transactions.
	As such, the adversary is able to distinguish diffused transactions and exclude them from the deanonymization set, thus improving precision.
	This means that diffused transactions do not actually contribute to mixing.
	On the other hand, proxy transactions are likely unknown to the adversary and are thus ideal for mixing.
	In particular, when receiving a proxy transaction from a node, the adversary is not able to determine whether it was created by such node or by one of its peers.
	Hence, the anonymity of the mixing set solely depends on the number of proxy transactions it contains.
	Therefore, we maximize anonymity by having nodes include only proxy transactions in their mixing set.
	
	In order to make nodes able to distinguish proxy transactions, we propagate them using a separate protocol message, called \texttt{PTX}.
	Instead, diffused transactions will be transmitted using the standard \texttt{TX} message.
	
	To further mitigate the risk of selecting an adversarial proxy, we only relay new transactions to outbound peers.
	This strategy is motivated by the fact that the adversary can control an arbitrary number of inbound connections.
	Instead, she has limited influence on the outbound peers, which are chosen at random among all reachable nodes in the network.
	
	Following the same reasoning, we exclude from the mixing set the transactions received from inbound peers.
	We do this to limit the ability of the adversary to track transactions in the mixing set of a node.
	In fact, an adversary being used as a proxy by a node, and controlling many inbound connections towards the same node, can track all the transactions she relays through these connections and then exclude them from the corresponding deanonymization set to improve precision.
	At the same time, to allow a correct propagation, we have nodes relay proxy transactions from inbound peers to other inbound peers.
	
	In summary, proxy transactions received from an outbound peer are relayed to another outbound peer, while proxy transactions received from an inbound peer are relayed to another inbound peer.
	We depict this scheme in Figure \ref{fig:clover}.
	The name Clover has been chosen to recall the four-way pattern of our protocol.
	\begin{figure}[tbp]
		\centerline{\includegraphics[width=0.3\columnwidth]{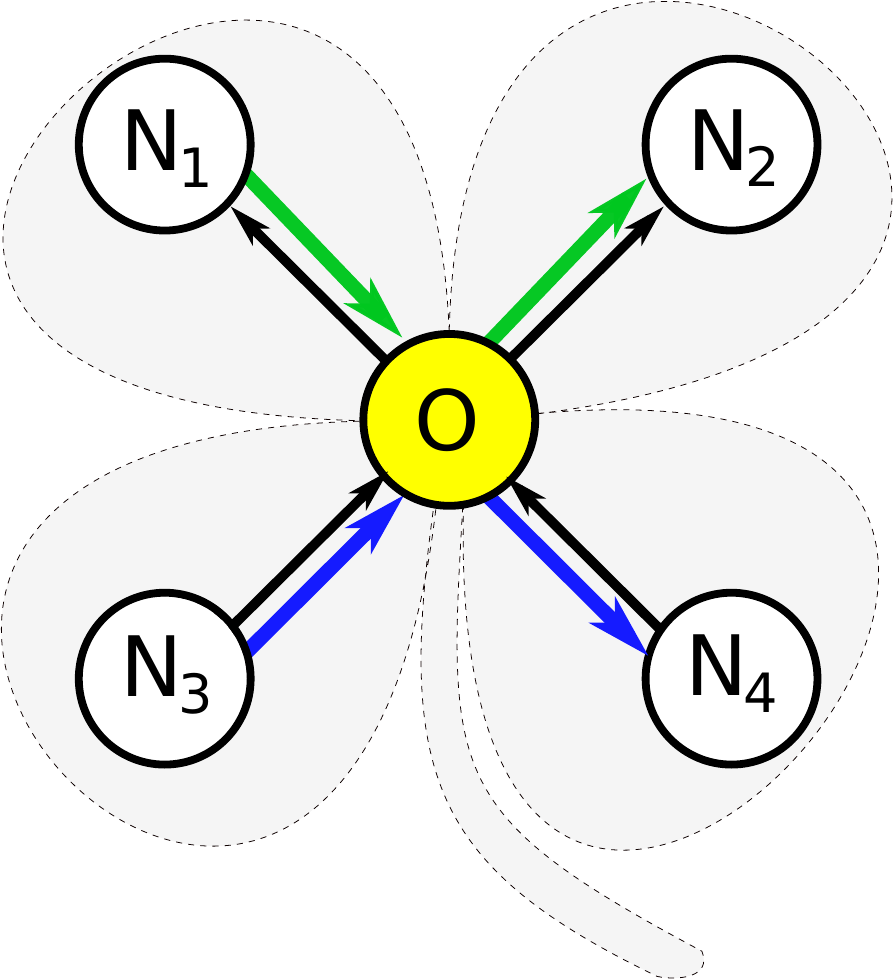}}
		\caption{The Clover relay protocol pattern: black arrows represent the direction of the connection; colored arrows represent relays of proxy transactions.}
		\label{fig:clover}
	\end{figure}
	
	\subsection{Protocol Design}
	In this section, we detail the Clover protocol design and describe the rules followed by network nodes.
	
	\subsubsection{Proxy transactions}
	We introduce a new protocol message \texttt{PTX}, used to propagate proxy transactions.
	The \texttt{PTX} message has the same structure as \texttt{TX} and is only used to mark a transaction in the proxying phase.
	Like the \texttt{TX} message, the \texttt{PTX} message contains the full transaction data.
	
	During the proxying phase, transactions are propagated directly from one node to another, without previously announcing them via \texttt{INV} messages.
	In fact, the standard three-step transmission is meant to avoid sending a transaction twice to the same node, which is likely to occur in gossip-like protocols.
	However, since proxy transactions are propagated over a linear path, nodes are rarely expected to receive them twice.
	Instead, the receiver of a proxy transaction is always expected not to know it.
	
	An upside of this strategy is that it allows us to substantially reduce the propagation delay introduced by the proxying phase.
	Specifically, since each relay operation only requires one message instead of three, the delay is reduced by one third.
	
	\subsubsection{Transaction Propagation}
	When a node creates a new transaction $tx$, it selects a random proxy among its outbound peers, and sends it $tx$ using a \texttt{PTX} message.
	This marks the beginning of the proxying phase for $tx$.
	
	During this phase, at each hop, the transaction is relayed (re-proxied) to another node, or broadcast via Diffusion.
	A node $N$ can receive a proxy transaction $tx$ from both outbound and inbound peers.
	When this occurs, $N$ behaves like follows:
	if $tx$ is received from an outbound peer, $N$ relays it to another outbound peer, chosen at random;
	if $tx$ is received from an inbound peer, $N$ broadcasts it with probability $p$, or relays it (with probability $1{-}p$) to an inbound peer, chosen at random.
	The probability $p$ is defined at a global level, and determines the average number of hops through which a transaction is relayed during the proxying phase.
	When a transaction gets broadcast, it enters the diffusing phase and follows the standard Diffusion protocol.
	
	Note that the broadcast step can only occur when $tx$ is received from inbound peers.
	In other words, proxy transactions received from outbound peers are always re-proxied.
	This choice allows nodes to maximize their mixing set by using all suitable transactions (i.e., those received from outbound peers).
	
	\subsubsection{Timeout}
	As the diffusion step is probabilistic, a transaction could be relayed too many times, producing an excessive delay in its propagation.
	
	To mitigate this risk, when a node proxies a transaction, it sets a timeout to verify that it gets correctly diffused.
	To do so, the node monitors \texttt{INV} messages coming from their outbound peers.
	When the timeout expires, the node checks if the majority of the outbound peers has advertised the transaction.
	If so, the transaction is considered as correctly diffused; otherwise, the node broadcasts the transaction using Diffusion.
	
	Again, we only consider outbound peers because they are the least likely to be controlled by the adversary.
	If we relied on inbound peers, an adversary, controlling the selected proxy and the majority of inbound peers, could trick the node by simply advertising the transaction from all such peers.
	
	Note that the timeout is applied to all proxied transactions, regardless of being new or relayed.
	This prevents the adversary from distinguishing the two cases, which could lead to deanonymization attacks.	
	A default value for the timeout can be defined after performing experiments on the network.
	However, each node might choose its own value, depending on the desired security level.
	
	\subsubsection{Clover Procedures}
	We first define the \textit{proxy} procedure as in Algorithm \ref{alg:proxy}.
	\vspace{-5pt}
	\begin{algorithm}[h]
		\caption{Proxy procedure}
		\label{alg:proxy}
		\begin{algorithmic}[1]
			\Statex ENV: Timeout $t$; NodeSet $OutPeers$
			\Procedure{Proxy}{Transaction $tx$, NodeSet $ProxySet$}
			\State $ProxySet := ProxySet \setminus \{tx.source\}$
			\State $proxy = pickRandomNode(ProxySet)$
			\State Send $\texttt{PTX}(tx)$ to $proxy$
			\State $wait(t)$
			\State confirmations := Count( \texttt{INV}($tx$) received from $OutPeers$ )
			\If {$confirmations < (|OutPeers|/2 + 1)$}
			\State \Call{Diffuse}{tx}
			\EndIf
			\EndProcedure
		\end{algorithmic}
	\end{algorithm}
	
	This procedure takes as inputs the transaction to be proxied ($tx$) and a set of peers ($ProxySet$) among which to choose the proxy.
	The procedure picks a random node from $ProxySet$ and sends it a \texttt{PTX} message containing $tx$.
	If the transaction is being relayed, its sender is excluded from the candidates (to avoid sending the message back to the sender).
	After sending the \texttt{PTX} message, a timeout $t$ is set.
	While $t$ is not expired, the node collects \texttt{INV} messages from its outbound peers.
	When $t$ expires, the node checks if the majority of outbound peers has announced $tx$.
	If so, the transaction is considered as diffused; otherwise, the transaction is broadcast.
	
	We then define the Clover propagation rules as in Algorithm \ref{alg:rules}.
	\vspace{-5pt}
	\begin{algorithm}[h]
		\caption{Clover Propagation Rules}
		\label{alg:rules}
		\begin{algorithmic}[1]
			\State ENV: Probability $p$; NodeSet $OutPeers$, $InPeers$
			\Procedure{Clover}{}
			\If{\text{ New(Transaction $tx$)}}
			\State \Call{Proxy}{tx, OutPeers}
			\EndIf
			
			\If{\text{Receive($\texttt{PTX}(tx)$)}}
			\If{$tx.source$ in $OutPeers$}
			\State \Call{Proxy}{tx, OutPeers}
			\Else
			\State $d = getRandProb()$
			\If{\text{$d < p$}}
			\State \Call{Diffuse}{tx}
			\Else 
			\State \Call{Proxy}{tx, InPeers}
			\EndIf
			\EndIf
			\EndIf
			\EndProcedure
		\end{algorithmic}
	\end{algorithm}
	
	When a node creates a new transaction $tx$, or receives \texttt{PTX}($tx$) from an outbound peer, it runs $Proxy(tx, OutPeers)$;
	if a \texttt{PTX}($tx$) message is received from an inbound peer, the node runs $Diffuse(tx)$ with probability $p$, and $Proxy(tx, InPeers)$ with probability $1{-}p$.
	
	\section{Discussion}
	\label{sec:analysis}
	In this section, we study the anonymity properties of the Clover protocol against an eavesdropper adversary using the first-spy estimator.
	
	\paragraph{Notation.}
	We use $R$ to denote the set of reachable nodes in the network and $S$ to denote the subset of reachable nodes controlled by the adversary (spies).
	Without loss of generality, we let $I$ and $O$ represent the average set of inbound and outbound peers of a node in the network.
	
	We use the term \textit{source} or \textit{origin} of a transaction to indicate the node that created it.
	Instead, we use the term \textit{sender} to indicate the node that sends a specific message.
	
	For the sake of readability, Table \ref{table:symbols} summarizes all parameters used in this section.

	\begin{table}[t]
		\centering
		\begin{tabular}{cl}
			\hline
			Parameter 		& Description 					  \\ 
			\hline
			$\mathcal{A}$   & The eavesdropper adversary \\
			$R$             & Set of reachable nodes          \\
			$O$             & Average set of outbound peers of a node \\
			$I$             & Average set of inbound peers of a node  \\
			
			$p$             & Probability of diffusion        \\
			$S$             & Set of nodes in R controlled by $\mathcal{A}$  \\
			
			$\rho_I$        & Average number of transactions received from an inbound peer		 \\
			$\sigma_I$      & Average number of transactions sent to an inbound peer      			 \\
			$\rho_O$        & Average number of transactions received from an outbound peer 			 \\
			$\sigma_O$      & Average number of transactions sent to an outbound peer      			 \\
			
			$\mathcal{M}$	& Average mixing set of a node \\
			$a$				& Adversarial outbound peers of a node \\
			$g$				& Average number of transactions generated by a node \\
			
			\hline
		\end{tabular}
		\caption{Parameter definitions}
		\label{table:symbols}
	\end{table}
	
	\subsection{Security}
	We consider an eavesdropper adversary $\mathcal{A}$ as described in Section \ref{sec:adversary}.
	As we will show, $\mathcal{A}$ gains no advantage by connecting to all nodes, nor by establishing multiple connections towards the same node.
	In fact, in our protocol, new transactions are only relayed through outbound connections, making the inbound peers controlled by $\mathcal{A}$ irrelevant to deanonymization.
	Instead, $\mathcal{A}$ gains precision by deploying more reachable nodes, as this increases her chances of being selected as a proxy node for new transactions.
	
	To analyze the anonymity properties of Clover, two important aspects must be studied first. 
	On the one hand, we need to know the probability of selecting an adversarial node as proxy for new transactions.
	On the other hand, we need to determine the size of the average mixing set for a single node.
	
	With these values, we can calculate the precision of $\mathcal{A}$ in deanonymizing proxy transactions, as well as its overall precision against all new transactions.
	Note that $\mathcal{A}$ will mainly target proxy transactions, as other transactions are unlikely to be announced by their source, in the Clover protocol.
	
	\paragraph{Proxy Selection}
	For the sake of simplicity, we assume each reachable node has the same probability of being selected as outbound peer when a new node joins the network\footnote{Although this assumption is theoretically sound, in the real Bitcoin network, well-established nodes tend to have more connections, especially compared to newly-joined nodes. This fact lowers the probability of connecting to the adversary, unless she is in control of a large portion of well-established nodes.}.
	
	Thus, we compute the probability of selecting an adversarial proxy for a single transaction as follows: 
	\begin{lemm}
		\label{lem:advprox}
		Let $R$ be the set of reachable nodes, and $S$ be the subset of nodes in $R$ controlled by the eavesdropper adversary $\mathcal{A}$,
		then the probability $P_\mathcal{A}$ of selecting an adversarial node as the proxy for a new transaction is:
		
		\begin{equation}
			P_\mathcal{A} = \frac{|S|}{|R|}.
		\end{equation}
	\end{lemm}
	\begin{proof}
		As each node establishes $|O|$ outbound connections, the probability of selecting a node in $R$ as an outbound peer is $|O|/|R|$. 
		As the adversary controls $|S|$ nodes in $R$, the probability of selecting a node in $S$ as an outbound peer is $|O||S|/|R|$.
		
		Since new transactions are sent to a random node in $O$, the probability of selecting a node in $S$ for a single new transaction is:
		\begin{equation}
			P_\mathcal{A} = \frac{1}{|O|}{\cdot}\frac{|O||S|}{|R|} = \frac{|S|}{|R|}.
		\end{equation}
	\end{proof}
	
	Therefore, in the current Bitcoin network, where $|R|\approx 10,000$, $\mathcal{A}$ would have $1/10000=0.0001$ probability of being selected as a proxy when controlling a single node.
	On the other hand, when controlling 1000 nodes (10\% of the reachable network) $\mathcal{A}$ would have $0.1$ probability of being selected for each new transaction sent in the network.
	
	Note that we are not taking into account other protective measures used by the Bitcoin client, such as the limitation in the number of peers from a single subnet, or the use of bucketing \cite{bucketing}.
	Since such measures are explicitly meant to reduce the probability of connecting to multiple adversarial nodes, it is likely that including these factors in the analysis would lower the value of $P_\mathcal{A}$.
	
	\subsubsection{Transaction Mixing}
	To ease the analysis, we study the mixing property of a node over a period of time $T$.
	However, as we will see, results are independent from this value.
	
	We want to calculate the average size of the mixing set of a node, which corresponds to the number of \texttt{PTX} messages received from outbound peers (i.e., nodes in $O$).
	In the following, we will use the word $transaction$ as a synonym of \texttt{PTX} message. 
	
	We use $\rho_I$ and $\sigma_I$ to denote the average number of transactions received from and sent to each node in $I$, respectively.
	Similarly, we use $\rho_O$ and $\sigma_O$ for transactions received from and sent to nodes in $O$.
	
	We study the size of the average mixing set $\mathcal{M}$ for a node having $a$ adversarial outbound peers.
	Note that, when all outbound peers are honest, the mixing set contains all transactions received from such peers (i.e., $|\mathcal{M}| = \rho_O|O|$).
	However, if $\mathcal{A}$ controls one or more outbound peers, the transactions received from these nodes are not useful for mixing (since they are known to $\mathcal{A}$).
	Therefore, in this case, the size of the average mixing set is $|\mathcal{M}| = \rho_O(|O|-a)$.
	
	Given the above, the following equation holds:
	\begin{lemm}
		\label{lem:mixing}
		Let $n$ be a generic node of the network, $g$ be the average number of transactions generated by $n$, $O$ be the set of outbound peers of $n$, $a$ be the subset of $O$ controlled by $\mathcal{A}$, and $p$ be the probability of Diffusion in Algorithm \ref{alg:proxy}.
		Then, the cardinality of the mixing set $\mathcal{M}$ for a node $n$ is:
		\begin{equation}
			|\mathcal{M}| = \frac{g(1-p)}{p}{\cdot}\frac{|O|-a}{|O|}.
		\end{equation}
	\end{lemm}
	\begin{proof}
		We consider the mixing set in the presence of $a$ adversarial nodes among the outbound peers: $|\mathcal{M}| = \rho_O(|O|-a)$.
		By definition, $\rho_O = \sigma_I$.
		Given the rules defined in Algorithm \ref{alg:rules},
		transactions received from nodes in $I$ ($\rho_I$) are relayed, with probability $1-p$, uniformly at random among nodes in $I$.
		Thus, we have: 
		\begin{equation}
			\sigma_I=(\rho_I|I|(1-p))/|I|=\rho_I(1-p).
		\end{equation}
		
		By definition, $\rho_I=\sigma_O$.
		Let us assume each node generates an average of $g$ transactions during $T$.
		Given that each node sends to nodes in $O$ all of its transactions along with those received by other nodes in $O$, we have:
		$\sigma_O = (g + \rho_O|O|)/|O|.$
		
		Given that $\rho_O=\sigma_I$ and $\rho_I=\sigma_O$, we have:
		\begin{equation}
			\sigma_O = \frac{(g + \sigma_O(1-p)|O|}{|O|} .
		\end{equation}
		
		Isolating $\sigma_O$, we get
		\begin{equation}
			\sigma_O = \frac{g}{|O|p} .
		\end{equation}	
		On the other hand, as $\rho_O=\sigma_I=\rho_I(1-p)=\sigma_O(1-p)$,
		we obtain:
		\begin{equation}	
			\begin{split}
				|\mathcal{M}| &= \rho_O(|O|-a)\\
				&= \sigma_O(1-p)(|O|-a)\\
				&= \frac{g}{|O|p}(1-p)(|O|-a)\\
				&= \frac{g(1-p)}{p}{\cdot}\frac{|O|-a}{|O|}.
			\end{split}
		\end{equation}
	\end{proof}
	
	Note that the size of the mixing set is inversely proportional to $p$.
	In fact, the smaller this value, the longer a transaction will be relayed before being diffused.
	In turn, the more a transaction is relayed, the more it contributes to the mixing of the other nodes.
	
	\subsubsection{Deanonymization precision}
	As previously mentioned, we expect $\mathcal{A}$ to mainly target proxy transactions, since it will be highly unlikely for her to receive diffused transactions from their source.
	Therefore, we first study the precision of $\mathcal{A}$ against the proxy transactions she receives.
	Then, we compute the overall accuracy considering all transactions.
	
	First, let us consider the precision against proxy transactions coming from a single node.
	Note that this only applies to nodes that opened a connection towards an adversarial peer.
	We assume $\mathcal{A}$ does not know incoming proxy transactions (although this might occasionally happen).
	As the first-spy estimator is used, each transaction is linked to the node that relayed it.
	
	Let $D_{proxy}$ be the average precision of $\mathcal{A}$ against proxy transactions coming from a single node. Then:
	\begin{lemm}
		\label{lem:dproxy}
		Let $n$ be a generic node of the network, $O$ be the set of its outbound peers, $a$ be number of peers in $O$ controlled by the eavesdropper adversary $\mathcal{A}$, and $p$ be the probability of Diffusion in Algorithm \ref{alg:proxy}.
		Then, the average precision of $\mathcal{A}$ against proxy transactions from a node $n$ is:
		\begin{equation}
			D_{proxy} = \frac{p}{1-\frac{a(1-p)}{|O|}} .
		\end{equation}
	\end{lemm}
	\begin{proof}
		We consider a node $n$ generating $g$ transactions, and being connected to $a$ outbound peers controlled by $\mathcal{A}$.
		As both new and relayed transactions are distributed among nodes in $O$, each such node receives on average $g/|O|$ new transactions plus $|\mathcal{M}|/|O|$ mixing transactions.
		Since $\mathcal{A}$ associates all transactions to $n$, she will get $g/|O|$ correct guesses over $(g+|\mathcal{M}|)$ transactions received.
		
		By Lemma \ref{lem:mixing}, we get:
		\begin{equation}
			\begin{split}
				D_{proxy}&=(g/|O|)/((g+|\mathcal{M}|)/|O|)\\
				&=g/(g+|\mathcal{M}|)\\
				&=g/(g+g\frac{1-p}{p}\frac{|O|-a}{|O|})\\
				&=\frac{p}{1-\frac{a(1-p)}{|O|}}.
			\end{split}
		\end{equation}
	\end{proof}

	To calculate the overall precision, we consider a network of $|N|$ nodes, $|R|$ of which are reachable.
	Let $D_{overall}$ be the overall precision of $\mathcal{A}$ against transactions generated by nodes in $N$. Then, the following equation holds:
	\begin{lemm}
		Let $R$ be the set of reachable nodes, and $S$ be subset of adversarial nodes in $R$.
		Then, the overall average precision of the eavesdropper adversary $\mathcal{A}$ against new transactions in the network is:
		\begin{equation}
			D_{overall} = \frac{|S|}{|R|}.
		\end{equation}
	\end{lemm}
	\begin{proof}
		Let us consider all transactions generated by nodes in $N$, that is $gN$.
		By Lemma \ref{lem:advprox}, each transaction is sent to an adversarial proxy with probability $|S|/|R|$.
		As such, $\mathcal{A}$ will receive $gN(|S|/|R|)$ transactions from their source (thus guessing them correctly).
		Dividing correct guesses over the total amount of transactions we have:
		
		\begin{equation}
			\frac{N\cdot\frac{|S|}{|R|}g}{gN} =  \frac{|S|}{|R|}.
		\end{equation}
	\end{proof}
	Therefore, the overall precision exclusively depends on the portion of reachable nodes controlled by $\mathcal{A}$.
	
	\subsection{Complexity and Efficiency}
	As it can be seen by the Clover procedures (Algorithms \ref{alg:proxy} and \ref{alg:rules}), the algorithm followed by network nodes has only plain instructions and if/then statements, without any loop.
	Since its complexity is constant (O(1)), Clover does not add any computational overhead to the Bitcoin protocol.
	
	Similarly, there is no expected overhead in the number of exchanged messages.
	In fact, like in the Diffusion protocol, transactions are propagated through all nodes of the network, without repetitions (although this can occasionally occur in Clover).
	Instead, since transactions in the proxying phase are transmitted directly (without previously announcing them), the total number of messages exchanged per node is expected to be lower than Diffusion.
	
	On the other hand, like other similar solutions, the Clover protocol introduces a delay in the broadcast of a transaction.
	Specifically, this delay is proportional to the number of hops through which transactions go during the proxying phase.
	
	In this respect, two factors must be considered: the number of messages needed for each hop, and the number of hops.
	
	\subsubsection{Hop delay}
	As described in Section \ref{sec:spreading}, in the Bitcoin protocol, each transaction propagation hop requires three messages: \texttt{INV}, \texttt{GETDATA}, and \texttt{TX}.
	This strategy is used in Diffusion to avoid sending transaction data to nodes that already have it.
	
	In Clover, this is not needed, since proxy transactions are normally unknown to the recipient.
	Instead, transaction data is transmitted directly with a single \texttt{PTX} message.
	Therefore, only one extra message is needed for each hop in the proxying phase.
	
	\subsubsection{Proxy hops}
	As previously stated, a higher number of relays during the proxying phase corresponds to a bigger mixing set for nodes in the network (and hence, better anonymity).
	Nevertheless, if this number is too high, it can cause an excessive propagation delay.
	Therefore, it is essential to choose a target value that seeks a compromise between efficiency and effectiveness.
	
	Note that the average number of hops directly depends on the probability $p$.
	In particular, the lower this value, the higher the number of hops.
	Therefore, we can choose $p$ to obtain a target number of hops ($h$).
	
	In Section \ref{sec:experiments}, we calculate the relation between $p$ and $h$, and experimentally evaluate the optimal target number of hops.

	\section{Comparison with State-of-the-Art Solutions}
	We compare Clover with other known anonymity-preserving propagation protocols.
	To the best of our knowledge, the only similar solutions proposed to date are Dandelion \cite{venkatakrishnan2017dandelion} (extended with Dandelion++ \cite{fanti2018dandelion++}), and the one proposed by Franzoni et al.~\cite{franzoni2020improving}.
	
	We review the main differences with Clover, and compare their complexity, efficiency, and security.
	
	\paragraph{Dandelion}
	This protocol, proposed by Fanti et al. in \cite{venkatakrishnan2017dandelion} and extended in \cite{fanti2018dandelion++}, is the first solution to have tried protecting transaction anonymity by breaking the symmetricity of propagation.
	Similar to Clover, Dandelion consists of two phases: a first lineal relay phase, called \textit{stem}, and a second broadcasting phase, called \textit{fluff}, where transactions are propagated using Diffusion.
	Transactions in the stem phase are relayed according to a propagation graph (a circle in Dandelion, and a 2-regular graph in Dandelion++), which is built by participating nodes prior to run the protocol.
	To that purpose each node selects one or two possible proxies (depending on the protocol version) uniformly at random among their outbound peers.
	By using a limited set of proxies, Dandelion aims at maximizing the mixing property since all transactions are relayed through the same path.
	
	Both in Dandelion and Clover, transactions are propagated only through outbound connections, minimizing the risk of proxying new transactions through adversarial nodes.
	However, Dandelion use transactions received from inbound peers for mixing, thus leaving space for the adversary to improve precision by controlling a large portion of inbound connections.
	For instance, let us consider the case in which an adversarial node is selected as a proxy by a victim;
	when this occurs, the adversary can open as many inbound connections as possible towards the victim to improve her chances of being used as inbound peer in the propagation graph.
	Whenever the adversary controls both one or more inbound peers and one or more outbound peers of such graph, she will be able to track all transaction used for mixing, and thus easily detect those generated by the victim.
	In Clover, we prevent this risk by only mixing with transactions relayed by other outbound peers.
	In other words, the adversary cannot gain any precision by opening inbound connections towards a victim.
	
	The use of a propagation graph in Dandelion not only increases the complexity of the protocol, but also introduces an additional attack vector for the adversary, which can gain precision by learning the topology of such graph.
	To avoid this risk, such graph has to be renewed every ten minutes, thus further increasing the complexity of the protocol.
	Clover avoid these issues by making use of all connected peers, and selecting proxies at random for each relay operation.
	
	With respect to the delay introduced by the initial proxying phase, Clover also outperforms Dandelion by transmitting transactions directly, without using the three-step relay process described in Section \ref{sec:txprop}.
	Roughly speaking, one hop in the stem phase of Dandelion introduces the delay of three hops in Clover.
	In other words, the delay introduced by each proxy hop in Clover is approximately one third than in Dandelion.
	Note that this allows for a longer proxying phase, which, in turn, means better anonymity properties (as the security of proxy transactions also depends on the average number of hops in such phase).
	
	Finally, a major limitation of Dandelion is to be only compatible with reachable nodes, which represent only the 10\% of the whole Bitcoin network.
	This is due to the fact that it requires nodes to have inbound connections.
	Conversely, Clover can also works when only outbound connections are available, thus being compatible with all nodes in the network.

	\paragraph{Reachability-dependent Anonymous Propagtion (ReAP)}
	Following an approach similar to Dandelion, Franzoni et al.~\cite{franzoni2020improving} proposed an alternative protocol that breaks the symmetry of transaction propagation by leveraging unreachable nodes.
	We call this protocol \textit{Reachability-dependent Anonymous Propagation}, or ReAP.
	
	In ReAP, transactions are again propagated in two phases, the first one of which have them relayed linearly through a sequence of proxy nodes.
	In the initial phase transactions are relayed through an alternate sequence of reachable and unreachable nodes.
	In particular, reachable nodes proxy transactions via unreachable nodes, and viceversa.
	This strategy has a twofold goal: to improve the involvement of unreachable nodes in transaction propagation, and to limit the ability of the adversary to observe the propagation pattern through the network.
	Their approach is based on the observation that the adversary is unable to open connections towards unreachable nodes, and hence cannot observe propagation through such nodes.
	
	Similar to Clover, ReAP does not require building a graph, is compatible with unreachable nodes, and minimizes delay by relaying transactions directly (i.e., without first announcing them) during the proxying phase.
	However, there are two major flaws in this protocol.
	First of all, reachable nodes require unreachable nodes to be connected in order to implement the protocol, which is not always the case. For instance, newly-joined nodes will likely have no inbound peers until their address is advertised to enough peers.
	Clover has no such limitation and can be readily be used by any node as soon as they connect to the network.
	
	The second major issue in ReAP lies in the ability of the adversary to open multiple connections from unreachable nodes towards reachable ones, which increases her chances to be selected as proxy for new transactions.
	Furthermore, this allows her to track a many transactions in the mixing set of the target, which, in turn, helps her improve precision in deanonymization.
	In Clover, we prevent this issue by having nodes mix only with transactions from other outbound peers, thus minimizing the ability of the adversary to track transactions in the mixing set of a target.
	
	Finally, the ReAP design has to deal with the complexity of determine the reachability of each node.
	In fact, this information is not explicit in the protocol and can only be inferred by probing the public listening address of a node.
	However, such address is not always advertised by nodes, making it hard to establish whether an inbound peer is reachable or not.
	Clover avoids such complexity by only differentiating between outbound and inbound connections, whose difference is well defined by the Bitcoin protocol and can be easily verified at any time.
	
	Finally, different from ReAP, we prove the anonymity guarantees of our protocol, both by formal analysis and experimental results.
	
	\paragraph{Summmary}
	In the previous paragraphs,  
	we compared Clover with state-of-the-art anonymous transaction propagation protocols.
	In particular, we discussed the complexity of their design, their scope, their security against the eavesdropper adversary, and the overhead introduced by the anonymity phase.
	
	In Table \ref{table:comparison}, we summarize this comparison. 
	As explained, both Dandelion/Dandelion++ and ReAP requires extra computation in order to enable the protocol; in contrast, Clover can be used without any previous operation.
	As for the delay introduced by the proxying phase, we saw how Clover and ReAP both minimize it to one extra message per hop.
	For what concerns the scope of the protocol, Clover is the only one that can be used by all nodes in the network, since Dandelion/Dandelion++ is only compatible with reachable nodes, while ReAP cannot be used by newly-joining reachable nodes.
	Finally, both Dandelion/Dandelion++ and ReAP allow the adversary to gain extra precision in deanonymization by means of side channels:
	in the first case, by learning the topology of the graph, and in the latter case, by connecting to a reachable node with multiple unreachable peers.
	On the contrary, an eavesdropper adversary can only gain precision against Clover by controlling a larger portion of reachable nodes.

	\begin{table}[t]
		\centering
		\begin{adjustbox}{width=0.8\textwidth}
			\begin{tabular}{|c|c|c|c|c|}
				\hline
				\textbf{Protocol} & \begin{tabular}[c]{@{}c@{}}\textbf{Extra}\\\textbf{Computation}\end{tabular} & \textbf{Overhead} & \textbf{Scope} & \begin{tabular}[c]{@{}c@{}}\textbf{Extra}\\\textbf{Precision}\end{tabular} \\
				\hline	
				\begin{tabular}[c]{@{}c@{}}\vspace{0pt}\\\textbf{Dandelion} \cite{venkatakrishnan2017dandelion} \\\textbf{Dandelion$++$} \cite{fanti2018dandelion++}\vspace{5pt}\end{tabular}  & 
				Yes    & 3 msg/hop & Reachable only & Yes \\ 
				\hline	
				\begin{tabular}[c]{@{}c@{}}\vspace{0pt}\\\textbf{ReAP} \cite{franzoni2020improving}\vspace{5pt}\end{tabular} & 
				Yes & 1 msg/hop & \begin{tabular}[c]{@{}c@{}}All except \\ new reachable nodes\end{tabular} & Yes \\ 
				\hline
				\begin{tabular}[c]{@{}c@{}}\vspace{0pt}\\\textbf{Clover}\vspace{5pt}\end{tabular} &
				No & 1 msg/hop & All & No \\ \hline
			\end{tabular}
		\end{adjustbox}
		\caption{Comparison between Clover, Dandelion, and ReaP}
		\label{table:comparison}
	\end{table}
	
	%
	%
	%
	
	In summary, Clover combines the strengths of previous solutions while mitigating their limitations and security risks.
	The resulting protocol has strong anonymity guarantees for all nodes in the network, with a simple design and minimum overhead.

	
	\section{Experimental Results}
	\label{sec:experiments}
	To evaluate the effectiveness of our protocol against an eavesdropper adversary, we performed a series of experiments in a simulated environment.
	In each series, we varied the portion of the network controlled by the adversary, so as to study the resilience of the protocol.
	
	We compare our results with those obtained using Diffusion in the same simulation setting.
	Our results show that Clover reduces the precision of the first-spy estimator up to ten times in the best case,
	while significantly increasing the cost of the attack for the adversary.
	
	\subsection{Proof of Concept and Simulation}
	We set our experiments in a private Bitcoin network using the reference client (Bitcoin Core 0.20), which we modified to implement the Clover protocol.
	For the experiments running the Diffusion protocol, we used the original implementation without modifications.
	
	\paragraph{Setting}
	To run the simulations we executed nodes in Regtest (private) mode using Docker containers.
	Each test was run on a network of 100 reachable nodes randomly connected to each other.
	In each simulation, we had nodes randomly generate transactions during 10 minutes.
	On average, in each run we generated approximately 300 transactions, with an average of 3 transactions per node.
	For each setting, we run 3 simulations and then computed the average.
	
	Being a simulated environment, our experimental setting might not fully represent the actual Bitcoin network.
	In particular, unlike our simulation, connections in the real network are not evenly distributed among nodes.
	Instead, stable nodes often maintain more connections than others, while newly-joined nodes typically require several hours before having a stable number of inbound peers.
	This might be an advantage for the adversary in the case it runs a well-known, stable node.
	However, in the real network, the adversary will be likely more limited than in our simulation, due to the fact that she needs to deploy several nodes to perform the deanonymization attack. 
	In other words, deanonymization results in our simulation are likely to be better than they would be in the real world.
	
	Unlike the real Bitcoin network, our simulation maintains a stable topology during the experiments.
	The stability of a controlled environment allows us to better evaluate the effectiveness of our propagation protocol against deanonymization attacks.
	Moreover, it allows a more meaningful comparison between Clover and Diffusion, since they can be tested in similar conditions, without depending on the randomness of the real network.
	
	Due to technical reasons, we also exclude unreachable nodes in our simulations.
	Note that this has no relevance for Clover, since we showed in Section \ref{sec:analysis} how precision exclusively depends on reachable nodes, but it might slightly affect the results for the Diffusion protocol.
	However, although unreachable nodes are theoretically relevant in Diffusion,
	studies showed how their involvement in the transaction propagation is extremely low compared to their number \cite{wang2017towards}, with as little as the 0.001\% of nodes sending transaction messages.
	We then consider this as a minor limitation.
	
	Overall, despite the differences between our simulated environment and the actual Bitcoin network, we believe our results are proper indicator of the security gains of Clover over Diffusion.

	\paragraph{Adversary}
	We varied the number of adversarial nodes from 1 to 30, corresponding to a range between 1\% and 30\% of the reachable network.
	These nodes are chosen randomly among those already deployed, so that they are well connected to the rest of the network.
	Note that this is the worst-case scenario since it assumes the adversary controls well-established nodes in the network.
	In addition, when testing against Diffusion, each adversarial node connects to all reachable peers (recall that this is not necessary when testing Clover, since it does not improve the precision of the adversary).
	
	All adversarial nodes log incoming \texttt{INV} and \texttt{PTX} messages.
	At the end of the simulation, these logs are merged and ordered by timestamp.
	Then, the first-spy estimator is applied, linking each transaction to the first peer that advertised or transmitted it to any of the adversarial nodes.
	
	\paragraph{Timeout}
	The diffusion timeout has been set to fit the local simulation environment, where transactions are produced and spread faster than the real network.
	In particular, verification timeout has been set to 1 minute.
	
	\subsection{Simulation Results}
	We evaluated precision against adversaries controlling 1\%, 2\%, 5\%, 10\%, 20\%, and 30\% of the network. 
	Each adversary is first tested against Diffusion, and then against Clover with broadcast probability $p$ equal to 0.2, 0.3, and 0.4.
	Overall precision is calculated as the average among all tests with a given adversarial power.
	Results are shown in Figure \ref{fig:results}.
	\begin{figure}[tbp]
		\centerline{\includegraphics[width=0.7\columnwidth]{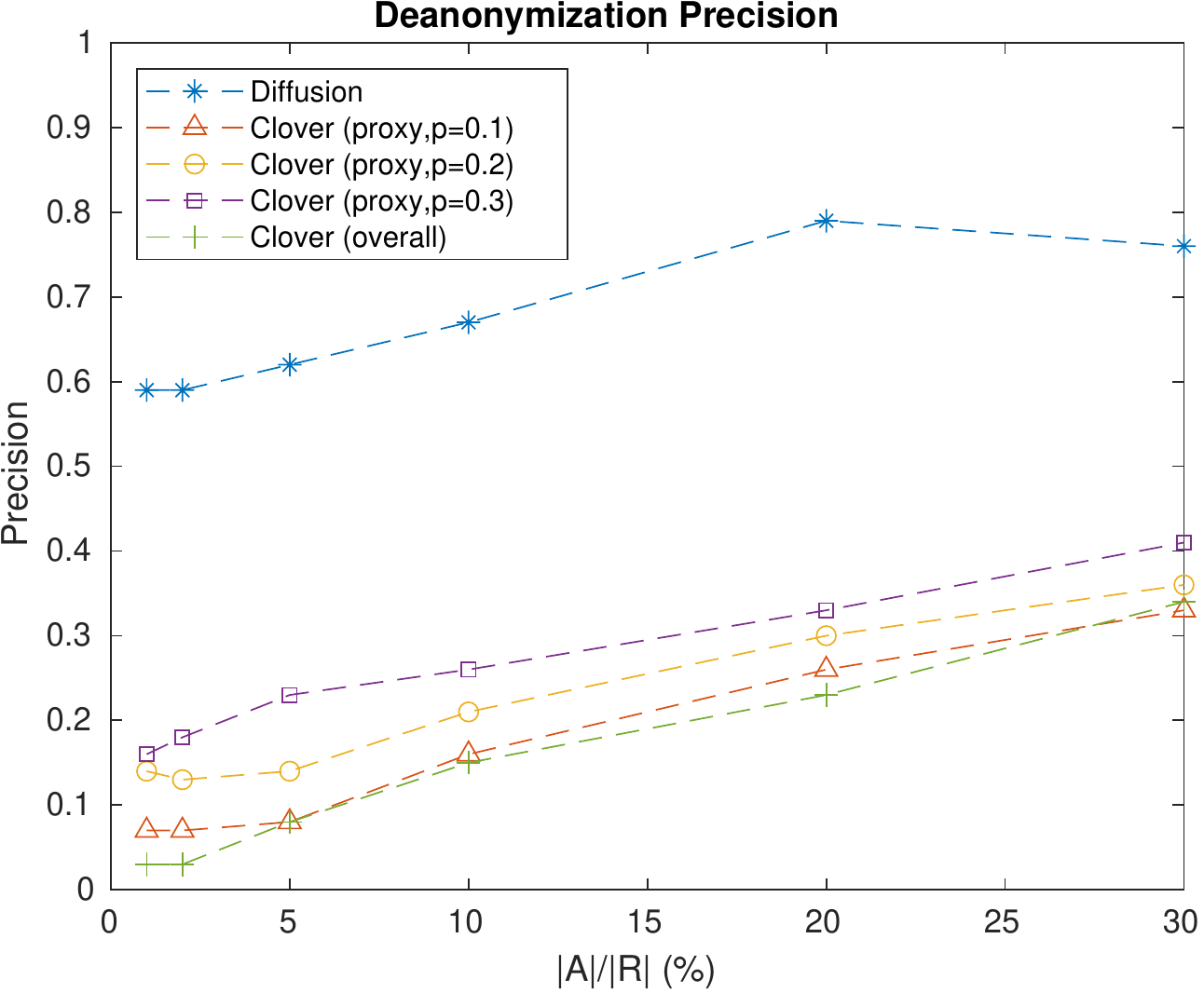}}
		\caption{Deanonymization precision against Clover}
		\label{fig:results}
	\end{figure}
	
	In our simulations, the precision of the adversary against Diffusion showed to be very high even controlling a small portion of the network.
	In particular, when controlling from 1\% to 5\% of the reachable network, the adversary had a precision as high as 0.6. 
	This value raises to 0.7 when the number of adversarial nodes reaches the 20\% of the network.
	
	On the other side, precision against Clover, although growing faster in the number of adversarial nodes, showed to be much lower than against Diffusion.
	Specifically, overall precision is from 10 times smaller (0.05), for adversaries controlling from 1\% to 5\% of the network, to 3 times smaller (0.33), for adversaries controlling from 10\% to 30\% of the network.
	
	For what concerns precision against proxy transactions,
	we have, as expected, better results for lower values of $p$.
	In particular, with $p=0.3$, precision ranges from 0.16 to 0.4, while, for $p=0.2$, the adversary showed an average precision of 0.14 when the controlling 1-5\% of the reachable nodes, and up to 0.35 when controlling 30\%.
	When setting $p=0.1$ precision against proxy transactions gets as low as the overall precision, indicating a near-optimum level of mixing.
	
	Notably, the precision of the adversary against Clover never exceeded that against Diffusion.
	This means that Clover against a strong adversary controlling 30\% of the network outperforms Diffusion against the weakest adversary controlling 1\% of the network.
	
	A major result of our experiments is that it shows how attacking Clover is substantially more expensive for the adversary (who need to deploy numerous nodes), compared to Diffusion, without even reaching the same levels of accuracy.
	
	\paragraph{Hops}
	According to our experiments, the average number of hops is inversely proportional to the probability $p$.
	In particular, we found the following relation to hold:
	\begin{equation}
		h \approx \frac{(1-p)}{0.15}  .
	\end{equation}
	For instance, with probability $p=0.1$, transactions are relayed through an average of 6 hops, during the proxying phase.

	\subsection{Comparison with Dandelion++}
	To further demonstrate the benefits of Clover, we experimentally compared its results against Dandelion++.
	To that purpose, we run the same set of experiments, in the same setting, using the official implementation of Dandelion++.
	We compare its results against Clover when the probability of diffusion is set to $p=0.1$, since this is also the value used by Dandelion++.
	Results are shown in Figure \ref{fig:results-dandelion}.
	
	As shown in the graphic, the overall precision of the eavesdropper adversary is comparable between Clover and Dandelion++.
	In contrast, precision against transactions in the proxying phase is visibly higher in Dandelion++.
	This is probably due to the fact that Clover better utilizes available transactions for the mixing property.
	In particular, while Clover distributes all incoming proxy transactions among all proxy nodes, Dandelion++ links each outbound proxy to a specific inbound peer. 
	Additionally, all new transactions are proxied through the same node (during one epoch).
	
	For the same reason, the results we obtained for Dandelion++ against proxy transactions were highly variable within a single setting.
	The irregularity of the corresponding line of the graphic reflects this variability.
	
	\begin{figure}[tbp]
		\centerline{\includegraphics[width=0.7\columnwidth]{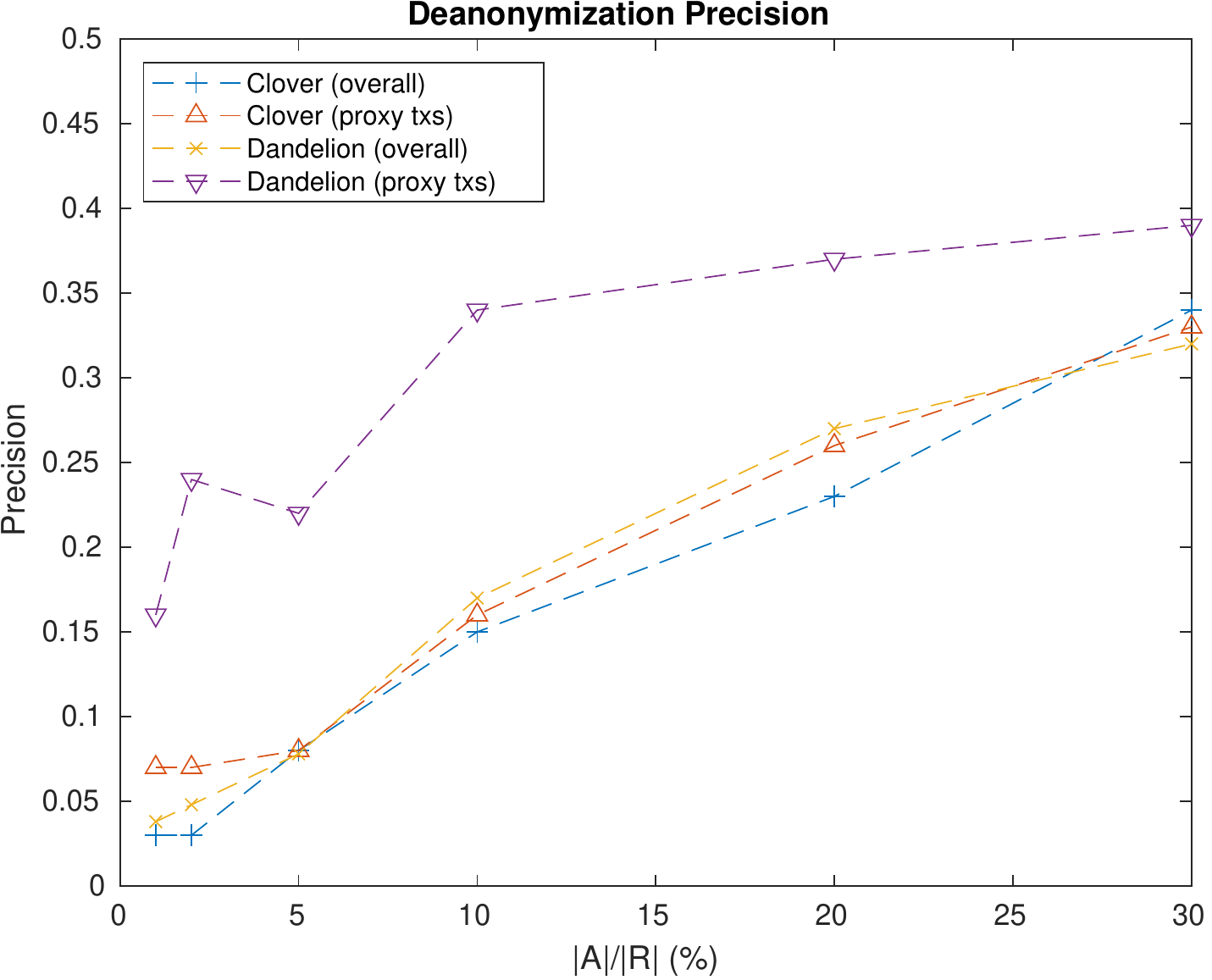}}
		\caption{Precision against Clover ($p{=}0.1$) and Dandelion}
		\label{fig:results-dandelion}
	\end{figure}
	
	In summary, when compared to Dandelion++, Clover shows a similar level of anonymity, but with better, and more stable, results for transactions in the proxying phase.
	
	\section{Related Work}
	Anonymity in Bitcoin has been widely addressed in research \cite{herrera2015research}.
	In particular, two major directions have been explored in relation to deanonymization.
	
	On the one side, there is blockchain analysis \cite{meiklejohn2013fistful, reid2013analysis}, which aims at linking Bitcoin addresses (and all related transactions) to real-world identities. 
	This is done by crossing publicly available information (e.g., known addresses or transactions, known services, ...) with address clustering: since all transactions are linked to each other, it is possible to trace coins throughout the whole blockchain.
	To prevent this kind of attack, users can use \textit{mix services}, which allow them to shuffle their coins with other users so as to prevent the possibility of tracing back coins in the blockchain.
	
	On the other side, there is traffic analysis, which aims at linking transactions to the IP address from which it originated, which would likely reveal the owner of the coins spent by such transactions.
	This is typically done by connecting to the whole network and monitoring transaction messages.
	Note that these attacks can go beyond the capacity of blockchain analysis, since they do not take into account Bitcoin addresses but only network packages.
	In other words, network analysis can deanonymize a transaction even when this is anonymized through a mix service.
	Although anonymity networks, like Tor or I2P, can be used to protect from such attacks, these services are not commonly used by Bitcoin users, and might even lead to other deanonymization attacks \cite{biryukov2015tor}.
	Therefore, network-level anonymity is still a major concern for Bitcoin users.
	In the following, we review the most relevant works related to network-level deanonymization.
	
	\paragraph{Network-Level Deanonymization}
	Dan Kaminsky~\cite{kaminsky2011black} is the first one to propose the general first-spy approach. 
	Based on the observation that nodes announce their transactions to all peers, he proposes to connect to all nodes and simply associates each transaction to the first node that announces it.
	
	Koshy et al.~\cite{koshy2014analysis} are among the first ones to apply the first-spy approach on the Bitcoin network. 
	In their experiment, they connect to all nodes during 5 months and analyze the relay patterns of each transaction. 
	Their results showed an accuracy of around 20\% using very conservative thresholds.
	
	In \cite{biryukov2019deanonymization}, Biryukov et al. propose a novel deanonymization technique targeting different cryptocurrency networks based on propagation analysis.
	Their approach is based on rumor centrality.
	According to their estimates, this technique is feasible even for low-budget adversaries.
	
	Differently from other works, which only apply to reachable nodes, Biryukov et al.~\cite{biryukov2014deanonymisation, biryukov2015tor} specifically target unreachable nodes and nodes using Tor. Their approach is based on fingerprinting techniques and reaches accuracy levels between 11\% and 60\%, depending on the stealthiness of the attacker.
	Since their technique only works during a single session, \cite{mastan2018new} propose a complementary technique that allows identifying unreachable nodes over multiple sessions.
	
	In \cite{fanti2017anonymity} and \cite{fanti2017deanonimization}, Fanti et al. theoretically analyze the anonymity properties of Trickle and Diffusion protocols against an eavesdropper adversary using first-spy and rumor-centrality-based estimators.
	Their results show that both protocols have poor anonymity guarantees and identify the symmetry of the propagation pattern as the core issue.

	\section{Conclusion}
	Transaction anonymity is considered an essential feature of cryptocurrencies.
	However, while great improvements have been made at the application level, the network level is still vulnerable to cheap and effective deanonymization attacks.
	Recent proposals have identified and addressed the issues in the propagation protocol that lead to such attacks.
	Nonetheless, the complexity of the proposed solutions might hinder their adoption in real networks.
	
	In this paper, we proposed an alternative approach to transaction propagation for the Bitcoin network, which adopts a simple design that eases its analysis and implementation.
	We theoretically studied its anonymity guarantees against powerful adversaries and experimentally evaluated its effectiveness through simulations, comparing results with the protocol currently used in Bitcoin.
	
	Our experimental results show that the deanonymization precision of the eavesdropper adversary adopting the first-spy estimator is up to 10 times smaller in the best case.
	We believe our solution can be easily adopted in real cryptocurrency networks and serve as a basis for future advances in the field.

	\printbibliography
\end{document}